\documentclass[a4paper,USenglish,cleveref,pdfa]{lipics-v2021}

\pdfoutput=1
\hideLIPIcs
\nolinenumbers

\usepackage{preamble}

\title{Faster Parameterized Broadcasting}
\titlerunning{Faster Parameterized Broadcasting}

\author{\'{E}douard Bonnet}
{CNRS, ENS de Lyon, Universit\'{e} Claude Bernard Lyon 1, LIP UMR 5668, Lyon, France
\and \url{http://perso.ens-lyon.fr/edouard.bonnet}}
{edouard.bonnet@ens-lyon.fr}
{https://orcid.org/0000-0002-1653-5822}
{Supported by the ANR project TWIN-WIDTH (ANR-21-CE48-0014-01).}

\author{Carl Feghali\footnote{We are deeply saddened to report that our colleague and friend, Carl Feghali, recently passed away. We greatly enjoyed working on this project with him.}}
{CNRS, ENS de Lyon, Universit\'{e} Claude Bernard Lyon 1, LIP UMR 5668, Lyon, France
\and \url{https://perso.ens-lyon.fr/carl.feghali/}}
{carl.feghali@ens-lyon.fr}
{https://orcid.org/0000-0001-6727-7213}
{}

\author{Manolis Vasilakis}
{Universit\'{e} Paris Dauphine -- PSL, CNRS UMR7243, LAMSADE, Paris, France
\and \url{https://manolisvasilakis.github.io/}}
{emmanouil.vasilakis@dauphine.eu}
{https://orcid.org/0000-0001-6505-2977}
{Supported by the ANR project S-EX-AP-PE-AL (ANR-21-CE48-0022) and
    the GDR ROD Mobility Scholarship.}

\authorrunning{\'{E}. Bonnet, C. Feghali, and M. Vasilakis}

\Copyright{\'{E}douard Bonnet, Carl Feghali, and Manolis Vasilakis}



\category{} 

\funding{}

\EventEditors{Neeldhara Misra and Magnus Wahlstr\"{o}m}
\EventNoEds{2}
\EventLongTitle{18th International Symposium on Parameterized and Exact Computation (IPEC 2023)}
\EventShortTitle{IPEC 2023}
\EventAcronym{IPEC}
\EventYear{2023}
\EventDate{September 6--8, 2023}
\EventLocation{Amsterdam, The Netherlands}
\EventLogo{}
\SeriesVolume{285}
\ArticleNo{22}

\begin{document}

\maketitle

\begin{abstract}
Given a connected graph $G$ and a source $s \in V(G)$,
what is the smallest number of rounds necessary for all
vertices of $G$ to receive a message initially only held by $s$,
where at each round every informed vertex passes the message to one of its neighbors?
This problem is called \textsc{Telephone Broadcast} and is suprisingly hard:
it remains NP-hard on cycles intersecting at a~single shared vertex, in particular,
graphs of pathwidth 2 with a linear feedback vertex set of size 1, as well as on graphs with treedepth at~most~6 [Egami et al.; MFCS '25].
Vertex cover number, vertex integrity, and distance to clique are among the few parameters for
which \textsc{Telephone Broadcast} is~fixed-parameter tractable.
There is a~$2^{\mathcal{O}(\mathrm{vc}^3)} n^{\mathcal{O}(1)}$-time algorithm parameterized by vertex cover number $\mathrm{vc}$ [Fomin, Fraigniaud, Golovach; TCS '24],
a~double-exponential algorithm parameterized by vertex integrity $\mathrm{vi}$,
and a~$2^{\mathcal{O}(k^2)} n^{\mathcal{O}(1)}$-time algorithm parameterized by distance to clique $k$ [Egami et al.; MFCS '25].

In this paper, we give improved parameterized algorithms for \textsc{Telephone Broadcast} with running times
$2^{\mathcal{O}(\mathrm{vc} \log \mathrm{vc})} n^{\mathcal{O}(1)}$,
$2^{\mathcal{O}(\mathrm{vi}^2 \log \mathrm{vi})} n^{\mathcal{O}(1)}$,
and $2^{\mathcal{O}(k \log k)} n^{\mathcal{O}(1)}$.
The main ingredient that makes our algorithms faster is a~Turing reduction to edge-weighted \textsc{$b$-Matching}.

\end{abstract}

\section{Introduction}

Broadcasting is one of the most basic primitives in communication networks:
a message initially held by one source has to reach every vertex of the network.
In the classical \emph{telephone model}~\cite{networks/HedetniemiHL88}, the
network is represented by a connected graph $G$, time is divided into
synchronous rounds, and in each round every informed vertex may transmit the
message to at most one uninformed neighbor.  For a source $s\in V(G)$, the
\emph{broadcast time} $b(G,s)$ is the minimum number of rounds needed to inform
all vertices.  If $|V(G)|=n$, then $\ceil{\log_2 n}\le b(G,s)\le n-1$,
and both bounds are tight: complete graphs realize the lower bound, while a
path whose source is an endpoint realizes the upper bound.  Equivalently, a~broadcast protocol can be represented by a spanning tree rooted at $s$ together
with an ordering of the children of every vertex. {\TB} then asks,
given a~connected graph $G$, a~source $s$, and an integer $t$, whether $b(G,s)\le t$.

The problem is old, natural, and algorithmically challenging.  Slater,
Cockayne, and Hedetniemi~\cite{siamcomp/SlaterCH81} initiated the systematic
study of the problem, proving that it is NP-complete in general but
polynomial-time solvable on trees.  The surrounding literature is extensive;
see, for instance, the survey of Hedetniemi, Hedetniemi, and
Liestman~\cite{networks/HedetniemiHL88}, the survey of Fraigniaud and
Lazard~\cite{dam/FraigniaudL94}, and the book of Hromkovic et
al.~\cite{sp/Hromkovic05}.  The computational difficulty already appears on
quite restricted graph classes: Jansen and M{\"{u}}ller~\cite{tcs/JansenM95}
proved NP-hardness on chordal graphs and grid graphs, among others.

The parameterized-complexity study of \TB\ was initiated by Fomin,
Fraigniaud, and Golovach~\cite{tcs/FominFG24}.  They gave a
$3^n n^{\bO(1)}$-time exact algorithm and proved fixed-parameter tractability
for feedback edge set number, vertex cover number, and the below-guarantee
parameter $n-t$.  Since a protocol of length $t$ can inform at most $2^t$
vertices, the exact algorithm also gives an FPT algorithm parameterized by
$t$, but with double-exponential dependence on~$t$.  This is essentially the
right order of magnitude under the Exponential-Time Hypothesis (ETH), by a consequence of the restricted
spanning-tree hardness of Papadimitriou and
Yannakakis~\cite{jacm/PapadimitriouY82}, as discussed by
Tale~\cite{tcs/Tale25}.

The recent structural picture is quite sharp on the negative side.  Tale
proved that \TB\ remains NP-complete on graphs of feedback vertex set number~1, and hence on graphs of treewidth~2~\cite{tcs/Tale25}; in fact, the same
work gives hardness for graphs at vertex-deletion distance~2 from a linear
forest.  Aminian, Kamali, Seyed-Javadi, and
Sumedha~\cite{icalp/AminianKJS25} proved NP-completeness on cactus graphs of
pathwidth~2, while also obtaining approximation results for cactus graphs and
bounded-pathwidth graphs.  Egami et al.~\cite{mfcs/EgamiGHKLMNOV025} conducted 
a~thorough study of the problem under various parameterizations and, apart from improving
(and significantly simplifying) the hardness result of~\cite{icalp/AminianKJS25},
pushed the hardness further to graphs of bounded treedepth, and complemented this with
positive results for vertex integrity, distance to clique, the combined
parameter $t+\tw$, and FPT approximation for dense parameters such as clique
cover and cluster vertex deletion.  Bringolf, Harutyunyan, Kamali, and
Seyed-Javadi~\cite{arxiv/BringolfHKS26} recently continued this line by
studying hardness and approximation for structured graph classes.

\subparagraph{Our Contribution.}
We improve the parameter dependence of the known positive results for vertex
cover number, vertex integrity, and distance to clique.
For vertex cover number $\vc$, in \cref{thm:algo:vc} we obtain an algorithm solving \TB\
in time $\vc^{\bO(\vc)}n^{\bO(1)}$, improving over the
$2^{\bO(\vc^3)}n^{\bO(1)}$-time algorithm of Fomin, Fraigniaud, and
Golovach~\cite{tcs/FominFG24}.
For vertex integrity, our algorithm runs in $\vi^{\bO(\vi^2)}n^{\bO(1)}$ time (\cref{thm:algo:vi}),
improving the double-exponential dependence of the
algorithm of Egami et al.~\cite{mfcs/EgamiGHKLMNOV025}.
Finally, we consider graphs whose vertex-deletion distance to clique is $k$,
and in \cref{thm:algo:dc} we obtain a $k^{\bO(k)}n^{\bO(1)}$-time algorithm,
improving the previous $2^{\bO(k^2)}n^{\bO(1)}$-time bound of Egami et al.~\cite{mfcs/EgamiGHKLMNOV025}.

On a high level, all three of our algorithms use the same organizing idea.
We first guess a small \emph{template} of the broadcast protocol, recording only the transmissions
that interact nontrivially with the relevant modulator.  For vertex cover, the
template contains the cover and the few independent-set vertices that relay
the message back into it.  For vertex integrity, the template contracts the few
components outside the modulator that are entered several times or that later
send the message back.  For distance to clique, the template keeps only the
clique vertices that interact with the modulator.  All other vertices or
components are left anonymous.  Once the template is fixed, the remaining task
is to realize the abstract slots and assign all leftover leaves to available
transmissions; we formulate this as an edge-weighted {\bMatching}
instance, which is polynomial-time solvable~\cite{book/Schrijver03}. We note that a
similar approach already appears in~\cite{tcs/FominFG24} (and subsequently in~\cite{mfcs/EgamiGHKLMNOV025}). 
However, a~refinement of it leads to significantly faster and simpler algorithms.

\subparagraph{Further Related Work.}
Several classical lines of research surround \TB.  One asks for exact
algorithms on special graph classes.  Besides trees~\cite{siamcomp/SlaterCH81},
polynomial-time algorithms are known for subclasses of cactus
graphs~\cite{jco/CevnikZ17}, unicyclic graphs~\cite{jco/HarutyunyanM08},
fully connected trees~\cite{join/GholamiHM23}, and generalized windmill
graphs, equivalently, graphs of twin-cover number~1~\cite{algorithms/AmbashankarH15,jgaa/Damaschke24}.
Another classical direction is extremal: one asks how sparse a graph can be
while still achieving the information-theoretic optimum $\ceil{\log_2 n}$;
Grigni and Peleg~\cite{siamdm/GrigniP91} gave tight bounds of this flavor.

The problem is also well studied from the viewpoint of approximation.  Ravi
gave an $\bO(\log^2 n/\log\log n)$-approximation~\cite{focs/Ravi94}, Kortsarz
and Peleg~\cite{siamdm/KortsarzP95} gave an algorithm producing a protocol of
length at most $2b(G,s)+\bO(\sqrt n)$, and Bar-Noy, Guha, Naor, and
Schieber~\cite{siamcomp/Bar-NoyGNS00} obtained an $\bO(\log n)$-approximation.
Elkin and Kortsarz~\cite{siamcomp/ElkinK05} gave a combinatorial logarithmic
approximation for the directed version, and later obtained the
$\bO(\log n/\log\log n)$ guarantee for telephone multicast~\cite{jcss/ElkinK06}.
Better approximation ratios are known for specific graph
classes~\cite{caldam/BhabakH15,join/BhabakH19,ciac/HarutyunyanH23,siamdm/KortsarzP95},
and inapproximability results are known as well~\cite{siamcomp/ElkinK05,approx/Schindelhauer00}.

Finally, many variants of the telephone model have been studied.  These include
line- and path-communication models~\cite{networks/Farley80,siamdm/KortsarzP95}
and multicast versions, where only a specified terminal set must be informed.
Most of the approximation algorithms above extend to multicast, and directed
or otherwise modified multicast variants have also been
considered~\cite{algorithmica/ElkinK06,algorithmica/HathcockKR26}.

\section{Preliminaries}\label{sec:preliminaries}

\subsection{Basic notation and broadcasting}

\subparagraph{Basic notation.}
We use standard graph notation~\cite{Diestel17} and assume familiarity with
the basic notions of parameterized complexity~\cite{books/CyganFKLMPPS15}.
All graphs are finite, simple, and undirected, unless explicitly stated otherwise.
For integers $x,y\in\Z$, let $[x,y]:=\setdef{z\in\Z}{x\le z\le y}$, and,
for a positive integer~$x$, let $[x]:=[1,x]$.
For a graph $G=(V,E)$ and a set $S\subseteq V$, we write $G[S]$ for the
subgraph induced by~$S$, and $G-S$ for $G[V\setminus S]$.
We denote the set of connected components of $G$ by $\cc(G)$.

\subparagraph{Telephone broadcast.}
Let $G$ be a connected graph and let $s\in V(G)$ be the source.
A~\emph{broadcast protocol} of length $t$ is a~spanning tree $T$ of $G$ (called its \emph{broadcast tree}),
rooted at~$s$, together with a time-stamp function $\lambda \colon E(T) \to [t]$.
Orient every edge of $T$ away from~$s$.
Say, $uv$ is oriented from the parent $u$ to the child~$v$.
Then $\lambda(uv)$ is the round in which $u$ informs~$v$.
Let $\tau_\lambda(s)=0$, and let $\tau_\lambda(v)$ be the label of the unique edge of~$T$
entering $v$ for every $v\ne s$.

The pair $(T,\lambda)$ is \emph{feasible} if $\tau_\lambda(u)<\lambda(uv)$ for every edge
oriented from $u$ to $v$, and if the outgoing edges of each vertex have pairwise distinct labels.
These are exactly the telephone constraints: a vertex transmits only after being informed,
and to at most one neighbor in each round.
Conversely, every feasible pair $(T,\lambda)$ defines a broadcast protocol.
The \emph{broadcast time} $b(G,s)$ is the minimum length of a broadcast
protocol for source~$s$ in~$G$.


\subsection{Templates}\label{sec:templates}

Our algorithms repeatedly guess a bounded part of a broadcast protocol.  We
use one common language for these guesses.  A \emph{template} records the named
vertices, a bounded number of abstract placeholders, and the transmissions
between them.  The remaining vertices or components are selected later,
by a {\bMatching} instance.

\begin{definition}[Template]\label{def:template}
    A \emph{template} is a tuple $\mathcal{T}=(H,s,S,Q,\lambda)$,
    where $S$ and $Q$ are disjoint finite sets, $s \in S$, $H=(V_H,A)$ is a
    loopless directed multigraph with $S\cup Q\subseteq V_H$, and
    $\lambda \colon A \to \N$.
    The vertices of $S$ are called \emph{named vertices};
    the vertices of $Q$ are called \emph{slots}.
    The vertices of $V_H\setminus(S\cup Q)$, if any, are called
    \emph{auxiliary vertices} and have no incoming arcs.
    Parallel arcs in~$H$ are allowed.
\end{definition}

When the remaining data are clear from the context, we sometimes refer to
$(H,\lambda)$ itself as the template.

\subparagraph{Meaning of the labels.}
An arc $e=(u,v) \in A$ with $\lambda(e)=\rho$ represents a transmission in round~$\rho$
from the object represented by~$u$ to the object represented by~$v$, unless $u$ is auxiliary.
Arcs leaving auxiliary vertices are interpreted by the algorithm that introduced them.
The labels determine the first labeled entry into every object that has an incoming arc.
Namely, for an object~$z \in S\cup Q$ we write $\tau_\lambda(z):= \min \setdef{\lambda(e)}{e\in A \text{ enters } z}$,
and we set $\tau_\lambda(s)=0$ for the source.
The consistency conditions below ensure that every vertex of $S\setminus\{s\}$ and every
slot $q \in Q$ has an incoming arc, so $\tau_\lambda$ gives its reception time.
Auxiliary vertices are only bookkeeping devices; they are not realized by vertices or components
of the input graph, and no reception time is associated with them.
In the distance-to-clique algorithm we use one such vertex~$\phi$ to denote a transmission supplied later from inside the clique.

\begin{figure}[!ht]
  \centering
  \begin{tikzpicture}[
      x=1cm,
      y=1cm,
      obj/.style={circle,draw,minimum size=7mm,inner sep=0pt,font=\scriptsize},
      slot/.style={circle,draw,fill=black!12,minimum size=7mm,inner sep=0pt,font=\scriptsize},
      edge/.style={->,semithick},
      stamp/.style={fill=white,inner sep=1pt,font=\scriptsize},
      setbox/.style={draw,rounded corners,inner sep=5pt},
      lab/.style={font=\small}
    ]

    \node[obj] (s) at (0,0) {$s$};
    \node[obj] (u1) at (-1.65,-1.25) {$u_1$};
    \node[slot] (q1) at (1.15,-1.25) {$q_1$};
    \node[obj] (u2) at (-1,-2.55) {$u_2$};
    \node[slot] (q2) at (2.75,-2.55) {$q_2$};

    \draw[edge] (s) -- node[stamp] {$1$} (u1);
    \draw[edge] (u1) -- node[stamp] {$3$} (q1);
    \draw[edge] (q1) -- node[stamp] {$4$} (u2);
    \draw[edge] (u2) -- node[stamp] {$5$} (q2);
    \draw[edge] (u2) to[bend left=20] node[stamp] {$7$} (q2);

    \node[setbox,fit=(s)(u1)(u2)] {};
    \node[setbox,fit=(q1)(q2)] {};
    \node[lab] at (-0.75,-3.4) {named vertices $S$};
    \node[lab] at (1.95,-3.4) {slots $Q$};
  \end{tikzpicture}
  \caption{An example of a template.}
  \label{fig:time-stamped-template}
\end{figure}
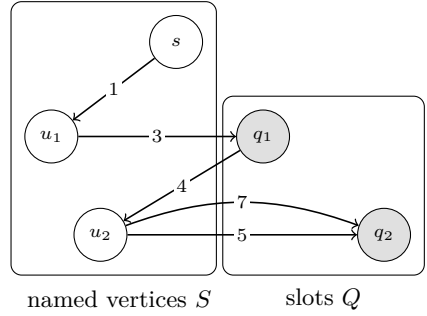

\subparagraph{Consistency.}
When a template~$(H,s,S,Q,\lambda)$ is checked against an input graph $G$, the named vertices are
actual vertices of~$G$, so $S \subseteq V(G)$, and $s \in S$ is the source.
We say that the template is \emph{consistent} with respect to~$G$ if it
satisfies the following common local checks:
\begin{itemize}
    \item $s$ has no incoming arc;
    \item each $u \in S \setminus \{s\}$ has exactly one incoming arc;
    \item all outgoing arcs of $u \in S$ have distinct labels, and each label
    is greater than $\tau_\lambda(u)$;
    \item every arc with both endpoints in $S$ is an edge of~$G$;
    \item every slot has at least one incoming arc;
    \item for every $q \in Q$, each outgoing arc $e$ from $q$ satisfies
    $\tau_\lambda(q) < \lambda(e)$.
\end{itemize}
Thus every represented transmission leaving an object is labeled later than
the first labeled entry into that object.
Each algorithm may impose additional structural requirements,
for example that labels belong to particular sets of rounds.

\subparagraph{Slots and realizations.}
Slots are placeholders.  In this paper, each slot is realized either by one
unnamed vertex or by one connected component of $G-S$; the algorithm at hand
specifies which type is allowed.  Thus a slot need not have a unique incoming
arc, and a~template need not be a~tree or even be acyclic.  Arcs incident with a
slot record only boundary transmissions.  The internal behavior of the realized
object and the choice of the actual endpoint inside it are checked when the
slot is realized.

Formally, a realization of $\mathcal{T}$ maps each slot to a concrete object in
$G-S$, with different slots mapped to disjoint objects.  It is valid if all
represented boundary transmissions not involving auxiliary vertices can be
implemented along edges of~$G$, and if every realized slot admits a local
protocol respecting the time-stamps of its incident arcs.
Auxiliary vertices, if present, are not realized.

\subparagraph{Extensions.}
An \emph{extension} of a realized template adds the transmissions not
represented in~$\mathcal{T}$, e.g., transmissions to independent-set vertices that do not
inform any vertices in the vertex cover, or transmissions inside a~clique.
A template is \emph{extendable} if some realization and some extension form
a valid broadcast protocol within the required number of rounds.
Our algorithms enumerate a bounded family of templates and test extendability by a~polynomial-time matching subroutine.

\subparagraph{Use in the algorithms.}
For vertex cover number, the template is a rooted tree: $S$ is the vertex cover and
the slots are the independent-set vertices that inform vertices of the cover,
each realized by a distinct independent-set vertex.  For vertex integrity,
$S$ is the modulator and the slots are the components of $G-S$ that either inform
vertices of the modulator or receive the message from the modulator multiple times,
each realized by a distinct connected component; here the template is a directed
multigraph and may have several arcs entering the same slot.
For distance to clique, $S$ is the modulator to the clique and the slots
are the clique vertices informing the modulator.  The unique auxiliary vertex~$\phi$
records slots whose first entry is supplied later from inside the clique.

\section{Vertex Cover Number}

We start with the parameterization by vertex cover number, denoted by~$\vc$.
Fomin, Fraigniaud, and Golovach~\cite{tcs/FominFG24} proved that {\TB} is
fixed-parameter tractable parameterized by~$\vc$, with running time
$2^{\bO(\vc^3)} n^{\bO(1)}$.  We improve this dependence to
$\vc^{\bO(\vc)} n^{\bO(1)}$.  Our algorithm keeps the main structural insight
of their proof: there is an optimal protocol in which the vertex cover is
informed very early.  We restate this lemma here and later use it to guess only
the part of the broadcast tree that is relevant to the vertex cover.

\begin{lemma}[{\cite[Lemma 4]{tcs/FominFG24}}]\label{lem:prefix-bp}
    For every connected graph $G$, vertex $s \in V(G)$, and vertex cover $S$ of~$G$,
    there is an optimal broadcast protocol for $(G,s)$ such that every vertex of $S$ receives the message by round $2|S|-1$.
\end{lemma}

We can now present the improved algorithm parameterized by vertex cover number.

\begin{theorem}\label{thm:algo:vc}
    {\TB} can be solved in time $\vc^{\bO(\vc)} n^{\bO(1)}$,
    where $\vc$ is the vertex cover number of the input graph.
\end{theorem}

\begin{proof}
    Let $(G,s,t)$ be an instance of {\TB}, where $G=(V,E)$ is connected and $|V|=n$.
    Assume that $t<n$, since otherwise the answer is trivially yes.
    We first compute in time $2^{\bO(\vc)} n^{\bO(1)}$ a~vertex cover $S_0$ of $G - s$ of size at most $\vc$~\cite{tcs/ChenKX10,stacs/HarrisN24}.
    We set $S:=S_0 \cup \{s\}$, $k:=|S| \le \vc+1$,
    and let $I:=V \setminus S$ and $L:=2k-1$.
    Note that $I$ is an independent set of~$G$, since $S$ is a~vertex cover of~$G$.
    By~\cref{lem:prefix-bp}, there is an optimal broadcast protocol for $G$
    with source~$s$ such that every vertex of $S$ receives the message by round~$L$.
    In the following, we search for such an optimal protocol.

    Fix such a protocol, and let $T$ be its broadcast tree.
    We first describe the bounded part of $T$ that we guess.
    A vertex of $I$ is \emph{marked} if it has at least one child in $S$.
    Since $I$ is independent, every marked vertex has its parent in $S$ and all its children in $S$.
    Moreover, every vertex of $S\setminus\{s\}$ has a unique parent.
    Thus, if we assign to every marked vertex one of its children in $S$,
    we obtain an injection from the set of marked vertices to $S\setminus\{s\}$.
    Hence there are at most $k-1$ marked vertices.

    The vertices of $S$ together with the marked vertices form a rooted
    subtree $T^\star$ of~$T$: on the path from any vertex of $S$ to the root,
    every vertex of $I$ that appears has a child in $S$,
    and is therefore marked.
    All vertices of this subtree receive the message by round~$L$:
    this is clear for the vertices of~$S$, and each marked vertex receives the message before informing one of its children in~$S$.

    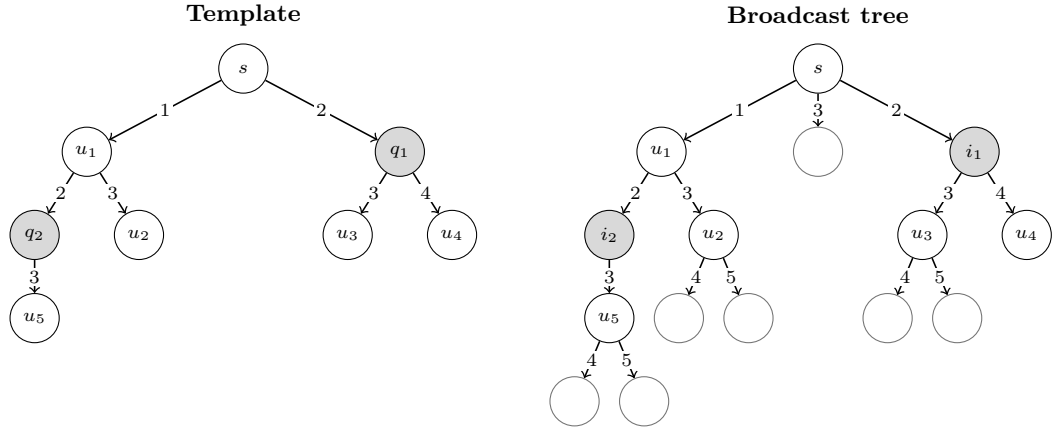
\begin{figure}[!ht]
        \centering
        \begin{tikzpicture}[
            x=1.15cm,
            y=1.1cm,
            vertex/.style={
            circle,
            draw,
            minimum size=6.5mm,
            inner sep=0pt,
            font=\scriptsize
            },
            cvertex/.style={vertex,fill=white},
            bvertex/.style={vertex,fill=black!15},
            other/.style={vertex,draw=black!55,fill=white},
            tree edge/.style={->,semithick},
            stamp/.style={fill=white,inner sep=1pt,font=\scriptsize},
            panel/.style={font=\small\bfseries}
        ]

        \begin{scope}
            \foreach \name/\x/\y/\lab in {s/0/0/s,u1/-1.8/-1/{u_1},u2/-1.2/-2/{u_2},c3/1.2/-2/{u_3},u4/2.4/-2/{u_4},u5/-2.4/-3/{u_5}}
            \node[cvertex] (L\name) at (\x,\y) {$\lab$};

            \foreach \name/\x/\y/\lab in {q1/1.8/-1/{q_1},q2/-2.4/-2/{q_2}}
            \node[bvertex] (L\name) at (\x,\y) {$\lab$};

            \foreach \u/\v/\t in {s/u1/1,s/q1/2,u1/q2/2,u1/u2/3,q1/c3/3,q1/u4/4,q2/u5/3}
            \draw[tree edge] (L\u) -- node[stamp] {$\t$} (L\v);

            \node[panel] at (0,.65) {Template};
        \end{scope}

        \begin{scope}[xshift=7.6cm]
            \foreach \name/\x/\y/\lab in {s/0/0/s,u1/-1.8/-1/{u_1},u2/-1.2/-2/{u_2},c3/1.2/-2/{u_3},u4/2.4/-2/{u_4},c5/-2.4/-3/{u_5}}
            \node[cvertex] (R\name) at (\x,\y) {$\lab$};

            \foreach \name/\x/\y/\lab in {q1/1.8/-1/{i_1},q2/-2.4/-2/{i_2}}
            \node[bvertex] (R\name) at (\x,\y) {$\lab$};

            \foreach \name/\x/\y in {x0/0/-1,x1/-1.6/-3,x2/-.8/-3,x3/.8/-3,x4/1.6/-3,x5/-2.8/-4,x6/-2.0/-4}
            \node[other] (R\name) at (\x,\y) {};

            \foreach \u/\v/\t in {s/u1/1,s/q1/2,u1/q2/2,u1/u2/3,q1/c3/3,s/x0/3,q1/u4/4,q2/c5/3,u2/x1/4,u2/x2/5,c3/x3/4,c3/x4/5,c5/x5/4,c5/x6/5}
            \draw[tree edge] (R\u) -- node[stamp] {$\t$} (R\v);

            \node[panel] at (0,.65) {Broadcast tree};
        \end{scope}
        \end{tikzpicture}
        \caption{Left: a vertex-cover template with
            $S = \{s,u_1,u_2,u_3,u_4,u_5\}$ and marked slots
            $Q = \{q_1,q_2\}$.
            Right: a broadcast tree extending the template, where
            $q_1,q_2$ are realized by $i_1,i_2\in I$, and the unlabeled leaves are
            unmarked vertices of~$I$.}
        \label{fig:template}
    \end{figure}

    We enumerate all possible \emph{templates} for such a subtree;
    see \cref{fig:template} for an illustration.
    For each $w \in \{0,\ldots,k-1\}$, we introduce a set $Q$ of $w$
    \emph{marked slots}.
    A marked slot is a~placeholder for a marked vertex of $I$ in the unknown
    broadcast tree; it will later be realized by a distinct vertex of~$I$.
    We enumerate all trees $H$ on vertex set $S \cup Q$, rooted at~$s$,
    such that every marked slot of $Q$ has its parent in $S$ and a non-empty set of children
    contained in~$S$.

    We orient $H$ away from~$s$ and enumerate a label $\lambda(e) \in [L]$ for every arc~$e$ of~$H$.
    The resulting labeled tree, with named vertices~$S$ and slots~$Q$, is a template $\mathcal{T}=(H,s,S,Q,\lambda)$ in the sense of \cref{def:template}.
    We discard the guess unless the consistency conditions of \cref{sec:templates} hold,
    with the additional constraint that for every $q \in Q$, all its outgoing arcs receive a distinct label.

    \begin{claim}
        The number of templates considered is $k^{\bO(k)}$.
    \end{claim}

    \begin{claimproof}
        Fix $w \in \{0,\ldots,k-1\}$ and let $m=k+w\le 2k$.
        There are at most $m^m\le (2k)^{2k}$ trees on the vertex set $S\cup Q$,
        and the root~$s$ determines their orientation.
        For each such tree, there are at most $L^{m-1}\le (2k)^{2k}$ choices of the labels,
        because $L=2k-1\le 2k$.
        The value of $w$ has only $k$ choices, and the consistency checks only discard guesses.
        Hence the total number of templates is at most $k\cdot (2k)^{4k}=k^{\bO(k)}$.
    \end{claimproof}

    Fix a~consistent template~$(H,\lambda)$.
    For notational convenience, write $\tau=\tau_\lambda$.
    For $q \in Q$, let $R(q)$ be the set of vertices of~$S$ adjacent to~$q$ in
    the underlying tree $H$.
    Thus $R(q)$ consists of the parent and the children of $q$ in the template.
    An actual vertex of $I$ can realize $q$ only if it is adjacent in $G$ to every vertex of $R(q)$.
    Let $I_q := \setdef{v \in I}{R(q) \subseteq N_G(v)}$ denote the set of such vertices.
    \Cref{fig:allowed-instantiation} illustrates in particular that,
    with the notation of the figures, $i_1 \in I_{q_1}$.

    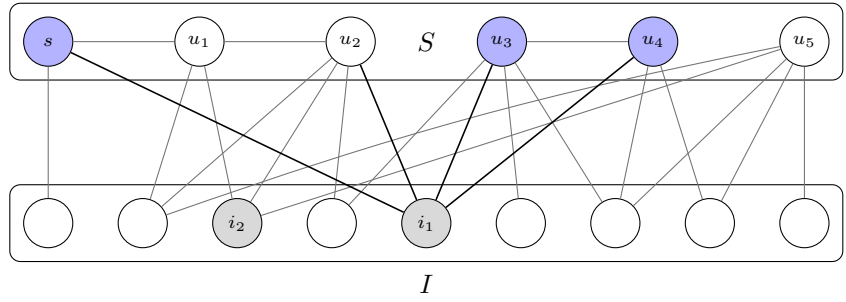
\begin{figure}[!ht]
    \centering
    \begin{tikzpicture}[
        x=1cm,
        y=1cm,
        vertex/.style={circle,draw,minimum size=6.5mm,inner sep=0pt,font=\scriptsize},
        cvertex/.style={vertex,fill=white},
        cvertexb/.style={vertex,fill=blue!30},
        ivertex/.style={vertex,fill=white},
        marked/.style={vertex,fill=black!15},
        edge/.style={semithick},
        extra/.style={thin,draw=black!55},
        rbedge/.style={semithick,draw=black},
        setbox/.style={draw,rounded corners,inner sep=5pt},
        setlabel/.style={font=\scriptsize,fill=white,inner sep=1pt}
        ]

        \foreach \name/\x/\lab in {c1/-3/{u_1},c2/-1/{u_2},c5/5/{u_5}}
        \node[cvertex] (\name) at (\x,0) {$\lab$};

        \foreach \name/\x/\lab in {s/-5/s,c3/1/{u_3},c4/3/{u_4}}
        \node[cvertexb] (\name) at (\x,0) {$\lab$};

        \foreach \name/\x in {u0/-5,u1/-3.75,u2/-1.25,u3/1.25,u4/2.5,u5/3.75,u6/5}
        \node[ivertex] (\name) at (\x,-2.4) {};

        \foreach \name/\x/\lab in {i2/-2.5/{i_2},i1/0/{i_1}}
        \node[marked] (\name) at (\x,-2.4) {$\lab$};

        \begin{scope}
        \node[setbox,fit=(s)(c1)(c2)(c3)(c4)(c5)] {};
        \node[setbox,fit=(u0)(i2)(u1)(u2)(i1)(u3)(u4)(u5)(u6)] {};
        \end{scope}

        \node at (0,0) {$S$} ;
        \node at (0,-3.2) {$I$} ;

        \foreach \u/\v in {s/c1,c1/c2,c3/c4,i2/c5,s/u0,c1/i2,c2/u1,c2/u2,c3/u3,c3/u4,c5/u5,c5/u6,c1/u1,c3/u2,c4/u4,c4/u5,c5/u4,c2/i2}{
        \draw[extra] (\u) -- (\v);}

        \draw[extra] (u1) to [bend left=4] (c5) ;

        \foreach \u/\v in {s/i1,i1/c3,i1/c4,i1/c2}{
        \draw[rbedge] (\u) -- (\v);}

    \end{tikzpicture}
    \caption{A graph admitting the right-hand side of~\cref{fig:template} as a~broadcast tree, with vertex cover $S$ and independent set $I$.
        Vertex $i_1 \in I$ realizes the marked slot $q_1 \in Q$, which is permitted since $R(q_1)=\{s,u_3,u_4\}$ (in blue) is a~subset of $N(i_1)=\{s,u_2,u_3,u_4\}$ (with incident thicker edges).}
    \label{fig:allowed-instantiation}
    \end{figure}

    Recall that in {\bMatching}, given an edge-weighted graph and a~map $\beta$ from its vertex set to the set of positive integers,
    we seek a~maximum-weight set $M$ of edges such that each vertex $u$ is incident with at most $\beta(u)$ edges of $M$ (regardless of their weight).
    We call $\beta(u)$ the capacity of~$u$.
    We create such an instance as follows.
    Its vertex set comprises $S \cup I$, and one additional vertex $w_q$ for every marked slot $q \in Q$.
    The capacities are
    \[
        \beta(v)=1 \quad\text{for every } v \in I,
        \qquad
        \beta(w_q)=1 \quad\text{for every } q \in Q,
    \]
    and, for every $u \in S$,
    \[
        \children(u)=\setdef{x \in V(H)}{x~\text{is a child of}~u~\text{in}~H},
        \qquad
        \beta(u)=t-\tau(u)-|\children(u)|.
    \]
    If $\beta(u)<0$ for some $u \in S$, we discard the template.
    If $\beta(u)=0$ for some $u \in S$, then we delete the vertex $u$ from the instance (as it has no capacity).
    We add an edge of weight~$2$ between $w_q$ and every vertex of $I_q$.
    We also add an edge of weight~$1$ between $u \in S$ and $v \in I$ whenever $uv \in E(G)$.
    We say that the template is \emph{extendable} if the resulting {\bMatching} instance has a
    feasible $b$-matching of weight at least $|I|+|Q|$.
    This can be checked in polynomial time, since {\bMatching} is polynomial-time solvable~\cite{book/Schrijver03}.

    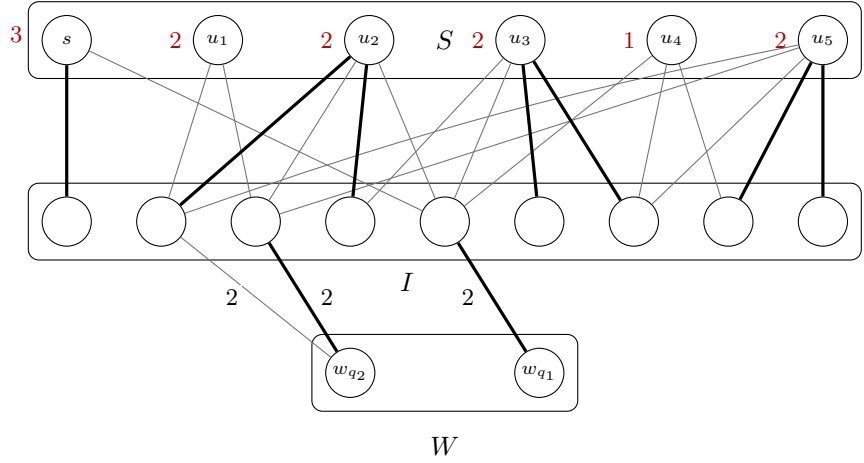
\begin{figure}[!ht]
        \centering
        \begin{tikzpicture}[
            x=1cm,
            y=1cm,
            vertex/.style={circle,draw,minimum size=6.5mm,inner sep=0pt,font=\scriptsize},
            cvertex/.style={vertex,fill=white},
            ivertex/.style={vertex,fill=white},
            wvertex/.style={vertex,fill=white},
            edge/.style={thin,draw=black!55},
            chosen/.style={very thick,draw=black},
            setbox/.style={draw,rounded corners,inner sep=5pt},
            caplabel/.style={font=\small,text=red!70!black,inner sep=1pt},
            wlabel/.style={font=\small,fill=white,inner sep=1.5pt}
        ]

        \foreach \name/\x/\lab in {s/-5/s,c1/-3/{u_1},c2/-1/{u_2},c3/1/{u_3},c4/3/{u_4},c5/5/{u_5}}
            \node[cvertex] (\name) at (\x,0) {$\lab$};

        \foreach \name/\x in {u0/-5,u1/-3.75,u2/-1.25,u3/1.25,u4/2.5,u5/3.75,u6/5,v2/-2.5,v1/0}
            \node[ivertex] (\name) at (\x,-2.4) {};

        \foreach \name/\x/\lab in {w2/-1.25/{w_{q_2}},w1/1.25/{w_{q_1}}}
            \node[wvertex] (\name) at (\x,-4.4) {$\lab$};

        \begin{scope}
            \node[setbox,fit=(s)(c1)(c2)(c3)(c4)(c5)] {};
            \node[setbox,fit=(u0)(u1)(v2)(u2)(v1)(u3)(u4)(u5)(u6)] {};
            \node[setbox,fit=(w2)(w1)] (Wbox) {};
        \end{scope}

        \node at (0,0) {$S$};
        \node at (-0.5,-3.2) {$I$};
        \node at ($(Wbox.south)+(0,-0.45)$) {$W$};

        \node[caplabel,anchor=east] at ([xshift=-6pt,yshift=2pt]s.west) {$3$};
        \node[caplabel,anchor=east] at ([xshift=-3pt]c1.west) {$2$};
        \node[caplabel,anchor=east] at ([xshift=-3pt]c2.west) {$2$};
        \node[caplabel,anchor=east] at ([xshift=-3pt]c3.west) {$2$};
        \node[caplabel,anchor=east] at ([xshift=-3pt]c4.west) {$1$};
        \node[caplabel,anchor=east] at ([xshift=-3pt]c5.west) {$2$};

        \foreach \u/\v in {s/u0,s/v1,c1/v2,c1/u1,c2/u1,c2/u2,c2/v1,c2/v2,c3/v1,c3/u2,c3/u3,c3/u4,c4/v1,c4/u4,c4/u5,c5/v2,c5/u4,c5/u5,c5/u6}
            \draw[edge] (\u) -- (\v);

        \draw[edge] (u1) to[bend left=4] (c5);

        \draw[edge] (w1) -- node[wlabel,pos=.5,left=5pt] {$2$} (v1);
        \draw[edge] (w2) -- node[wlabel,pos=.5,right=5pt] {$2$} (v2);
        \draw[edge] (w2) -- node[wlabel,pos=.5,left=5pt] {$2$} (u1);

        \foreach \u/\v in {w1/v1,w2/v2,s/u0,c2/u1,c2/u2,c3/u3,c3/u4,c5/u5,c5/u6}
            \draw[chosen] (\u) -- (\v);

        \end{tikzpicture}
        \caption{The edge-weighted {\bMatching} instance associated with the template and the graph of~\cref{fig:template,fig:allowed-instantiation}.
        Capacities are in red and edge weights are in black.
        Implicit capacities and weights have value~1.
        The bold edges represent the $b$-matching for the broadcast tree of~\cref{fig:template}.}
        \label{fig:b-matching-reduction}
    \end{figure}

    \begin{claim}
        $b(G,s) \le t$ if and only if at least one consistent template is extendable.
    \end{claim}

    \begin{claimproof}
        Suppose first that $b(G,s) \le t$.
        By the discussion above, there is a protocol in at most $t$ rounds in which every vertex of $S$ receives the message by round $L$.
        Take its template $(H,\lambda)$, with labels given by the reception
        times in the protocol.
        We construct a feasible $b$-matching $M$.
        For every marked slot $q \in Q$, let $v_q \in I$ be the actual vertex realizing it, and put the edge $w_q v_q$ into $M$.
        This edge exists because $v_q \in I_q$.
        Every vertex of $I$ that is not marked is a leaf of the broadcast tree, because $I$ is independent and all its potential children lie in $S$.
        If such a leaf $v \in I$ has parent $u \in S$, put the edge $uv$ into $M$.

        By construction, $M$ contains (at~most) one edge incident to each vertex of $I$ and each vertex $w_q$.
        So their capacity is respected.
        After receiving the message at round $\tau(u)$, the vertex $u \in S$ has $t-\tau(u)$ transmissions available before the deadline.
        The template already accounts for exactly $|\children(u)|$ of them.
        Hence the number of additional leaf children of $u$ in $I$ is at most $\beta(u)$, and the capacity of $u$ is respected.
        Finally, the weight of $M$ is
        \[
            2|Q|+(|I|-|Q|)=|I|+|Q|,
        \]
        so the template is extendable.

        Conversely, suppose that a consistent template is extendable.
        Let $M$ be a feasible \mbox{$b$-matching} of~weight at~least $|I|+|Q|$.
        Since every vertex of $I$ has capacity $1$, each vertex of $I$ is incident with at most one edge of $M$.
        Such an edge has weight $2$ only if it is incident with some $w_q$, and otherwise has weight $1$.
        Let $m_2$ be the number of weight-$2$ edges in $M$.
        Then $m_2 \le |Q|$, because each $w_q$ has capacity $1$.
        The remaining matched vertices of $I$ contribute weight at most $1$ each, and there are at most $|I|-m_2$ of them.
        Therefore every feasible $b$-matching has weight at most
        \[
            2m_2+(|I|-m_2)=|I|+m_2 \le |I|+|Q|.
        \]
        It follows that $M$ has weight exactly $|I|+|Q|$,
        and equality in the above upper bound implies that $m_2=|Q|$, every vertex of $I$ is saturated, and every vertex $w_q$ is saturated.

        We construct a broadcast protocol as follows.
        Start with the guessed template.
        For each edge $w_q v \in M$, replace $q \in Q$ by the actual vertex $v$.
        This is possible because $v \in I_q$, so $v$ is adjacent in $G$ to the parent and all children prescribed for $q$ in the template.
        The capacity-$1$ constraints on the vertices of $I$ ensure that these vertices are distinct and are not used again.
        Every remaining vertex of $I$ is saturated by a unique edge $uv \in M$ with $u \in S$;
        make each such vertex $v$ a leaf child of $u$.
        This gives a~spanning tree of $G$ rooted at $s$.

        We schedule every template arc from a parent to a vertex $x$ at round $\tau(x)$.
        For each $u \in S$, the number of added leaf children assigned to $u$ is at most $\beta(u)$,
        and there are exactly $\beta(u)$ rounds among $\tau(u)+1,\ldots,t$ not already occupied by template transmissions from $u$.
        Schedule these leaf transmissions arbitrarily in those free rounds.
        All transmissions are along edges of $G$, no vertex transmits before receiving the message, and no vertex transmits to two vertices in the same round.
        Therefore every vertex receives the message by round $t$.
    \end{claimproof}

    We enumerate $k^{\bO(k)}$ templates, and for each of them solve a polynomial-size {\bMatching} instance.
    Since $k\le \vc+1$, the total running time is $\vc^{\bO(\vc)} n^{\bO(1)}$.
\end{proof}

\section{Vertex Integrity}

In this section we consider the parameterization by \emph{vertex integrity}, denoted by $\vi$.
The vertex integrity of a graph is the smallest value~$k$ such that there exists a vertex
set $S$ whose deletion leaves only connected components of size at most $k-|S|$.
Vertex integrity generalizes vertex cover number in a straightforward way,
since it is upper-bounded by the vertex cover number plus one.

Egami et al.~\cite{mfcs/EgamiGHKLMNOV025} proved that {\TB} is
fixed-parameter tractable parameterized by vertex integrity.
The dependence on the parameter in their proof is double-exponential;%
\footnote{Their algorithm considers two cases.  In the first one, the stated
double-exponential dependence can already be improved to
$2^{\bO(\vi^3)}n^{\bO(1)}$ by using their Theorem~12 as a subroutine.  The
second case, however, still takes time $2^{\bO(\vi^5)}n^{\bO(1)}$.}
we improve it to $\vi^{\bO(\vi^2)}n^{\bO(1)}$.

At a high level, the algorithm guesses only the part of the protocol where a
component of $G-S$ interacts nontrivially with the modulator~$S$.
There are few components that send the message back to~$S$, and these components,
together with the vertices of~$S$, can be assumed to receive the message
only at the very start or at the very end of the protocol~\cite[Lemma 8]{mfcs/EgamiGHKLMNOV025}.
Almost every other component receives the message from~$S$ only once;
the only exceptions occur near the end of the protocol~\cite[Lemma 7]{mfcs/EgamiGHKLMNOV025}.
The template records exactly these bounded interactions,
while all remaining components are assigned later by a {\MaxMatching} instance
(equivalently, an edge-weighted {\bMatching} instance with unit capacities).

\begin{theorem}\label{thm:algo:vi}
    {\TB} can be solved in time $\vi^{\bO(\vi^2)} n^{\bO(1)}$,
    where $\vi$ is the vertex integrity of the input graph.
\end{theorem}

\begin{proof}
    Let $(G,s,t)$ be an instance of {\TB}, where $G=(V,E)$ is connected and $|V|=n$.
    Assume that $t<n$, since otherwise the answer is trivially yes.
    Using the algorithm of Drange et al.~\cite{algorithmica/DrangeDH16},
    we compute in time $\vi^{\bO(\vi)} n^{\bO(1)}$ a set $S_0 \subseteq V$ such
    that every component of $G-(S_0 \cup \{s\})$ has at most $\vi - |S_0|$ vertices.
    Let $S := S_0 \cup \{s\}$, $\mathcal{D} := \cc(G-S)$,
    $d := \max_{D\in\mathcal{D}}|V(D)|$, and
    $k := |S|+d \le \vi+1$.

    \proofsubparagraph{Notation.}
    For a protocol with broadcast tree~$T$, an \emph{entry} into a component
    $D\in\mathcal{D}$ is an edge of~$T$ directed from~$S$ to~$D$.  A component is a
    \emph{relay component} if some edge of~$T$ is directed from it to~$S$, and
    otherwise it is a \emph{leaf component}.  A leaf component is
    \emph{exceptional} if it has at least two entries from~$S$.  Relay and
    exceptional components are called \emph{special}.

    Let $\Lambda := [1,\min\{t,(2k-3)k\}] \cup [\max\{1,t-2k+1\},t]$
    and $J := [\max\{1,t-k+2\},t]$.
    Consequently, we have that $|\Lambda|=\bO(k^2)$ and $|J|\le k-1$.
    Using the results of~\cite{mfcs/EgamiGHKLMNOV025} we obtain the following.

    \begin{claim}\label{clm:vi:canonical-protocol}
        If $b(G,s)\le t$, then there is a broadcast protocol~$P$ of length at
        most~$t$, with broadcast tree~$T$, such that the following hold.
        \begin{enumerate}
            \item Every vertex of~$S\setminus\{s\}$ and every vertex in a
            relay component of~$P$ is informed in a round of~$\Lambda$.
            \item Every exceptional component of~$P$ has all its entries in
            rounds of~$J$.
        \end{enumerate}
    \end{claim}

    \begin{claimproof}
        Let $\ell=b(G,s)\le t$.
        Following the notation of~\cite{mfcs/EgamiGHKLMNOV025}, we say that a
        protocol is \emph{$S$-lazy} if, among all optimal broadcast protocols, it
        minimizes the total number of message forwardings from~$S$.
        By~\cite[Lemma 8]{mfcs/EgamiGHKLMNOV025}, there is an $S$-lazy optimal
        protocol~$P_0$ of length~$\ell$ in which every vertex of
        $S\setminus\{s\}$ and every vertex in a relay component is informed
        either during the first $(2k-3)k$ rounds or during the last $2k$ rounds.
        Moreover, \cite[Lemma 7]{mfcs/EgamiGHKLMNOV025} implies that every leaf component
        of~$P_0$ with at least two entries from~$S$ has all its entries after
        round $\ell-k+1$.

        To obtain $P$, we shift the time-stamps of the last $2k$ rounds of~$P_0$ to
        occupy rounds $[t-2k+1, t]$ instead, leaving all earlier time-stamps unchanged.
        This only delays a~suffix of the protocol.
        Hence feasibility is preserved: earlier and shifted time-stamps remain disjoint,
        and every transmission delayed in this way still occurs after its sender is informed.
        The resulting protocol has length at most~$t$.

        The first $(2k-3)k$ rounds of~$P_0$ either remain unchanged or are
        shifted into the last $2k$ rounds with respect to~$t$, while the last
        $2k$ rounds of~$P_0$ are shifted into the last $2k$ rounds with respect
        to~$t$.  This proves the first item.  Since the last $k-1$ rounds of
        $P_0$ are contained in the last~$2k$ ones, every entry into a leaf component with at
        least two entries is shifted into~$J$ as well, proving the second item.
    \end{claimproof}

    \begin{claim}\label{clm:special-comp}
        Let $P$ be a protocol with broadcast tree~$T$ satisfying the two
        properties of \cref{clm:vi:canonical-protocol}.  Then $P$ has at most
        $|S|-1$ relay components and at most $|S|(k-1)/2$ exceptional components.
    \end{claim}

    \begin{claimproof}
        Each relay component informs at least one vertex of~$S\setminus\{s\}$,
        and distinct relay components inform distinct such vertices, since
        every vertex has a unique parent in~$T$.  Hence there are at most
        $|S|-1$ relay components.

        By the assumed second property, every entry into an exceptional
        component occurs in a round of~$J$.  Since $T$ is a tree, these entries
        inform distinct vertices of the component.  Each exceptional component
        has at least two entries, while the vertices of~$S$ make at most
        $|S|\abs{J}\le |S|(k-1)$ transmissions in these rounds.
        Thus there are at most $|S|(k-1)/2$ exceptional components.
    \end{claimproof}

    Thus only $\bO(k^2)$ special components need to be represented explicitly in a broadcast protocol satisfying the two properties of \cref{clm:vi:canonical-protocol}.
    Every other component is a leaf component entered exactly once from~$S$.

    \proofsubparagraph{Templates.}
    We now describe the bounded object that we guess, in the same spirit as
    the template in the proof of \cref{thm:algo:vc}.  The vertices of~$S$ are
    kept explicitly.  Each special component is contracted to one formal
    vertex, called a \emph{slot}.  A slot is a \emph{relay slot} or an
    \emph{exceptional slot}, according to the type of the component that it
    represents.

    Formally, this is a template $\mathcal{T}=(H,s,S,Q,\lambda)$ in the sense of
    \cref{def:template}, where $H=(S\cup Q,A)$.  We write
    $Q_{\mathrm{rel}}\subseteq Q$ for the relay slots and
    $Q_{\mathrm{exc}}=Q\setminus Q_{\mathrm{rel}}$ for the exceptional slots.
    The arcs in~$A$ record the transmissions involving~$S$ and the special
    components, and are only of the following forms, where we additionally record
    the set of allowed labels for each type of arc:
    \[
        S \xrightarrow{\Lambda} S,\qquad
        S \xrightarrow{\Lambda} Q_{\mathrm{rel}},\qquad
        Q_{\mathrm{rel}} \xrightarrow{\Lambda} S,\qquad
        S \xrightarrow{J} Q_{\mathrm{exc}}.
    \]
    An arc incident with a slot records only the endpoint in~$S$ and the round
    of the transmission; the endpoint inside the represented component is
    chosen later, when the slot is realized.

    Unlike the vertex-cover template, this guessed object need not be a tree.
    A special component may receive several entries from~$S$, and a relay
    component may also send the message back to~$S$.
    \Cref{fig:vi-template} illustrates the guessed object and how it is
    realized by actual components of~$G-S$.

    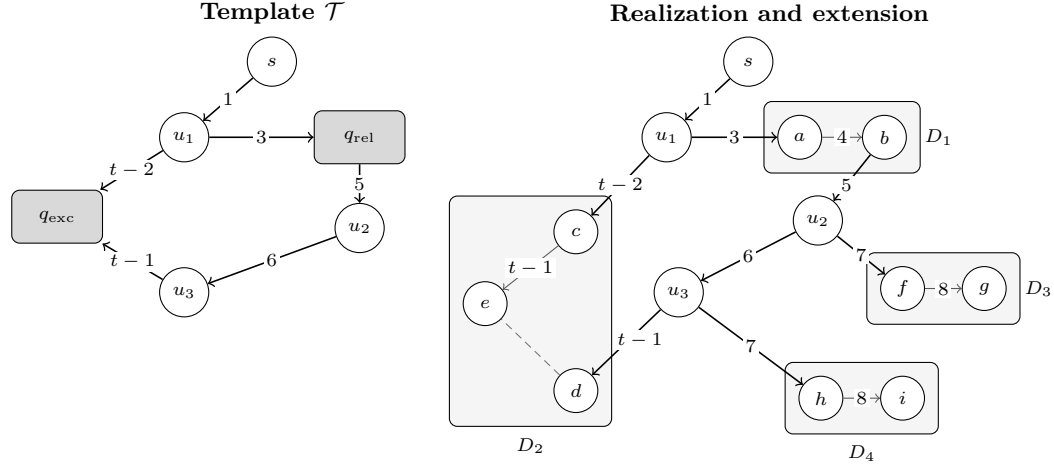
\begin{figure}[!ht]
        \centering
        \begin{tikzpicture}[
            x=.8cm,
            y=1cm,
            named/.style={circle,draw,fill=white,minimum size=6.5mm,inner sep=0pt,font=\scriptsize},
            slot/.style={rounded corners=3pt,draw,fill=black!15,minimum width=12mm,minimum height=7mm,inner sep=1pt,font=\scriptsize},
            cvertex/.style={circle,draw,fill=white,minimum size=5.8mm,inner sep=0pt,font=\scriptsize},
            compbox/.style={draw,rounded corners=3pt,fill=black!4,inner sep=5pt},
            tree edge/.style={->,semithick},
            internal/.style={->,thin,draw=black!65},
            graph edge/.style={thin,densely dashed,draw=black!55},
            stamp/.style={fill=white,inner sep=1pt,font=\scriptsize},
            panel/.style={font=\small\bfseries},
            complabel/.style={font=\scriptsize}
        ]

        \begin{scope}
            \foreach \name/\x/\y/\lab in {s/0/0/s,u1/-1.45/-1/{u_1},u2/1.45/-2.2/{u_2},u3/-1.45/-3.05/{u_3}}
                \node[named] (L\name) at (\x,\y) {$\lab$};

            \node[slot] (Lq1) at (1.45,-1) {$q_{\mathrm{rel}}$};
            \node[slot] (Lq2) at (-3.55,-2.05) {$q_{\mathrm{exc}}$};

            \draw[tree edge] (Ls) -- node[stamp] {$1$} (Lu1);
            \draw[tree edge] (Lu1) -- node[stamp] {$3$} (Lq1);
            \draw[tree edge] (Lq1) -- node[stamp] {$5$} (Lu2);
            \draw[tree edge] (Lu2) -- node[stamp] {$6$} (Lu3);
            \draw[tree edge] (Lu1) -- node[stamp] {$t-2$} (Lq2);
            \draw[tree edge] (Lu3) -- node[stamp] {$t-1$} (Lq2);

            \node[panel] at (0,.65) {Template $\mathcal T$};
        \end{scope}

        \begin{scope}[xshift=6.3cm]
            \foreach \name/\x/\y/\lab in {s/0/0/s,u1/-1.35/-1/{u_1},u2/1.15/-2.1/{u_2},u3/-1.15/-3.05/{u_3}}
                \node[named] (R\name) at (\x,\y) {$\lab$};

            \foreach \name/\x/\y/\lab in {a/.85/-1/{a},b/2.25/-1/{b},c/-2.85/-2.25/{c},d/-2.85/-4.35/{d},e/-4.35/-3.2/{e},f/2.55/-3/{f},g/3.9/-3/{g},h/1.2/-4.45/{h},i/2.55/-4.45/{i}}
                \node[cvertex] (R\name) at (\x,\y) {$\lab$};

            \begin{scope}[on background layer]
                \node[compbox,fit=(Ra)(Rb)] (RDrel) {};
                \node[compbox,fit=(Rc)(Rd)(Re)] (RDexc) {};
                \node[compbox,fit=(Rf)(Rg)] (RDleaf) {};
                \node[compbox,fit=(Rh)(Ri)] (RDleafb) {};
            \end{scope}
            \node[complabel] at ($(RDrel.east)+(.3,0)$) {$D_1$};
            \node[complabel] at ($(RDexc.south)+(0,-.25)$) {$D_2$};
            \node[complabel] at ($(RDleaf.east)+(0.3,0)$) {$D_3$};
            \node[complabel] at ($(RDleafb.south)+(0,-.25)$) {$D_4$};

            \draw[tree edge] (Rs) -- node[stamp] {$1$} (Ru1);
            \draw[tree edge] (Ru1) -- node[stamp] {$3$} (Ra);
            \draw[internal] (Ra) -- node[stamp] {$4$} (Rb);
            \draw[tree edge] (Rb) -- node[stamp,pos=.65] {$5$} (Ru2);
            \draw[tree edge] (Ru2) -- node[stamp] {$6$} (Ru3);
            \draw[tree edge] (Ru1) -- node[stamp] {$t-2$} (Rc);
            \draw[tree edge] (Ru3) -- node[stamp,pos=.3,below] {$t-1$} (Rd);
            \draw[internal] (Rc) -- node[stamp] {$t-1$} (Re);
            \draw[graph edge] (Rd) -- (Re);
            \draw[tree edge] (Ru2) -- node[stamp] {$7$} (Rf);
            \draw[internal] (Rf) -- node[stamp] {$8$} (Rg);
            \draw[tree edge] (Ru3) -- node[stamp] {$7$} (Rh);
            \draw[internal] (Rh) -- node[stamp] {$8$} (Ri);

            \node[panel] at (.35,.65) {Realization and extension};
        \end{scope}
        \end{tikzpicture}
        \caption{Left: a vertex-integrity template with
            $S = \{s,u_1,u_2,u_3\}$ and slots $Q = \{q_{\mathrm{rel}},q_{\mathrm{exc}}\}$.
            Right: a broadcast tree extending the template, where
            $q_{\mathrm{rel}},q_{\mathrm{exc}}$ are realized by $D_1,D_2 \in \mathcal{D}$,
            and $D_3,D_4 \in \mathcal{D}$ are leaf components of the protocol.}
        \label{fig:vi-template}
    \end{figure}

    We enumerate templates as follows.
    We use at most $|S|-1$ relay slots and at most $|S|(k-1)/2$ exceptional slots,
    as guaranteed by \cref{clm:special-comp}.
    For these slots we enumerate arcs of the allowed forms above, with labels from the relevant sets of rounds.
    The next claim shows that it is enough to consider guesses with $\bO(k^2)$ arcs.
    Here we call a (vertex-integrity) template \emph{consistent} if it satisfies the common
    consistency checks of \cref{sec:templates} and the following additional conditions:
    \begin{itemize}
        \item every relay slot has at least one incoming arc from~$S$ and at
        least one outgoing arc to~$S$;
        \item every exceptional slot has at least two incoming arcs from~$S$;
        \item no slot has more than $d$ incoming arcs.\footnote{Recall that $d$ is the largest number of vertices among the connected components of $G-S$.}
    \end{itemize}

    \begin{claim}
        The number of templates considered is $k^{\bO(k^2)}$.
    \end{claim}
    \begin{claimproof}
        The number of slots is $\bO(k^2)$, and hence every template has
        $\bO(k^2)$ vertices.  In a~consistent template, relay slots contribute
        at most $(|S|-1)d$ incoming arcs from~$S$ and at most $|S|-1$
        outgoing arcs to~$S$.  Exceptional slots contribute at most
        $|S|\abs{J}\le |S|(k-1)$ incoming arcs, because each vertex of
        $S$ transmits at most once in each round of~$J$.  Finally, there are at
        most $|S|-1$ arcs with both endpoints in~$S$.  Thus the number of
        arcs is $\bO(k^2)$.

        For each arc we have $\bO(k^2)$ choices for its tail, $\bO(k^2)$
        choices for its head, and $\bO(k^2)$ choices for its label.
        Therefore the total number of guesses is $k^{\bO(k^2)}$,
        while the consistency checks only discard guesses.
    \end{claimproof}

    Fix a consistent template.
    A component $D\in\mathcal{D}$ \emph{realizes} a
    slot $q\in Q$ if each arc incident with~$q$ can be implemented, with its
    prescribed direction and time-stamp, by choosing an endpoint in~$D$ adjacent
    to the corresponding vertex of~$S$, and by adding transmissions inside~$D$,
    so that every vertex of~$D$ receives the message exactly once by
    round~$t$, and every vertex of~$D$ transmits only after it is informed and
    at most once per round.  For a relay slot we additionally require that all
    vertices of~$D$ receive in rounds of~$\Lambda$.  No further restriction is
    needed for an exceptional slot: all its entries occur in~$J$, so every
    valid realization completing by round~$t$ receives all vertices of~$D$ in
    rounds of~$J$.
    Since a slot has $\bO(k)$ incident arcs and $|V(D)|\le d\le k$,
    whether $D$ realizes it can be tested in $k^{\bO(k)}$ time by exhaustively
    guessing the corresponding local protocol.

    For every component $D \in \mathcal{D}$ and every $v \in V(D)$, compute $b(D,v)$.
    Since $|V(D)|\le d \le k$, all these values can be computed in total time
    $2^{\bO(k)} n^{\bO(1)}$ by the $3^n n^{\bO(1)}$-time algorithm of Fomin et al.~\cite{tcs/FominFG24}.

    \proofsubparagraph{Extending a template.}

    Fix a consistent template~$(H,\lambda)$, and, for notational convenience, write $\tau=\tau_\lambda$.
    Let $Q$ be its set of slots.
    The components assigned to its slots account for every interaction with $S$
    that is not a single entry into a leaf component.
    Thus, as for the unmarked independent-set vertices in the proof of
    \cref{thm:algo:vc}, every remaining component only has to be assigned one
    unused transmission from~$S$.

    We construct a {\MaxMatching} instance.
    Its vertex set contains one vertex $v_D$ for every component $D \in \mathcal{D}$ and one vertex $w_q$ for every slot $q \in Q$.
    We add an edge of weight $2$ between $v_D$ and $w_q$ exactly when $D$ realizes $q$.
    For every $u \in S$ and every round $\rho \in \{\tau(u)+1,\dots,t\}$
    in which no arc of~$H$ leaves $u$ with label~$\rho$, create a vertex $x_{u,\rho}$.
    We add an edge of weight~$1$ between $v_D$ and $x_{u,\rho}$ if there is a~vertex
    $v \in N_G(u) \cap V(D)$ such that $\rho+b(D,v) \le t$.
    The template is \emph{extendable} if this instance has a~matching of
    weight at least
    \[
        |\mathcal{D}|+|Q|.
    \]
    The weight-$2$ edges pay for covering both a component and a slot, while
    the weight-$1$ edges pay only for covering a component by an ordinary one-entry broadcast from~$S$.
    The instance has polynomial size, since $t<n$, and
    {\MaxMatching} is polynomial-time solvable~\cite[Section~11]{books/KorteV18}.

    \begin{figure}[!ht]
        \centering
        \begin{tikzpicture}[
            x=1cm,
            y=1cm,
            vertex/.style={circle,draw,minimum size=7mm,inner sep=0pt,font=\scriptsize},
            xvertex/.style={vertex,fill=white},
            dvertex/.style={vertex,fill=white},
            wvertex/.style={vertex,fill=white},
            edge/.style={thin,draw=black!55},
            chosen/.style={very thick,draw=black},
            setbox/.style={draw,rounded corners=3pt,inner sep=5pt},
            setlabel/.style={font=\small},
            wlabel/.style={font=\small,fill=white,inner sep=1.5pt}
        ]

        \foreach \name/\x/\lab in {x1/-1.4/{x_{u_2,7}},x2/1.4/{x_{u_3,7}}}
            \node[xvertex] (\name) at (\x,0) {$\lab$};

        \foreach \name/\x/\lab in {d1/-4.2/{v_{D_1}},d2/-1.4/{v_{D_2}},d3/1.4/{v_{D_3}},d4/4.2/{v_{D_4}}}
            \node[dvertex] (\name) at (\x,-2.15) {$\lab$};

        \foreach \name/\x/\lab in {w1/-4.2/{w_{q_{\mathrm{rel}}}},w2/-1.4/{w_{q_{\mathrm{exc}}}}}
            \node[wvertex] (\name) at (\x,-4.25) {$\lab$};

        \begin{scope}
            \node[setbox,fit=(x1)(x2)] (Xbox) {};
            \node[setbox,fit=(d1)(d2)(d3)(d4)] (Dbox) {};
            \node[setbox,fit=(w1)(w2)] (Wbox) {};
        \end{scope}

        \node[setlabel] at ($(Xbox.north)+(0,.35)$) {unused transmissions from $S$};
        \node[setlabel,anchor=east] at ($(Dbox.west)+(-.25,0)$) {$\mathcal{D}$};
        \node[setlabel,anchor=east] at ($(Wbox.west)+(-.25,0)$) {$W$};

        \foreach \u/\v in {w1/d1,w2/d2,w1/d2,x1/d2,x1/d3,x2/d3,x2/d4}
            \draw[edge] (\u) -- (\v);

        \foreach \u/\v in {w1/d1,w2/d2,x1/d3,x2/d4}
            \draw[chosen] (\u) -- (\v);

        \node[wlabel,anchor=west] at ($(w1)!.52!(d1)+(.13,0)$) {$2$};
        \node[wlabel,anchor=west] at ($(w2)!.52!(d2)+(.13,0)$) {$2$};
        \node[wlabel,anchor=west] at ($(w1)!.55!(d2)+(.13,-.05)$) {$2$};

        \end{tikzpicture}
        \caption{Part of a {\MaxMatching} instance associated with the template of~\cref{fig:vi-template}.
        Edge weights are in black; implicit weights have value~1.
        Bold edges highlight a~solution that corresponds to the realization of~\cref{fig:vi-template}.}
        \label{fig:vi-matching}
    \end{figure}
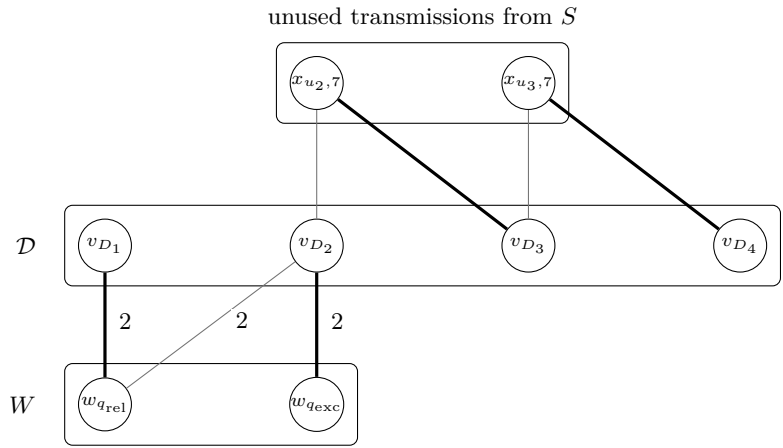

    \begin{claim}
        If $b(G,s)\le t$, then an extendable template exists.
    \end{claim}
    \begin{claimproof}
        Suppose that $b(G,s) \le t$ and take the protocol~$P$ satisfying the two
        properties of \cref{clm:vi:canonical-protocol}.
        Introduce a relay slot for every relay component and an exceptional slot for every exceptional component,
        and form~$H$ by recording precisely the transmissions incident with~$S$ in the
        restriction of~$P$ to these components and~$S$.
        The label ranges follow from \cref{clm:vi:canonical-protocol}, and the
        counting arguments above ensure that this template is enumerated.
        Each special component realizes its corresponding slot.
        Add the corresponding weight-$2$ edge to a~matching~$M$.

        Every component not represented by a slot is a leaf component entered exactly once.
        If such a component~$D$ is entered from~$u \in S$ through~$v \in V(D)$
        at round~$\rho$, then the restriction of~$P$ to~$D$ implies
        $\rho+b(D,v) \le t$.  Add the edge between~$v_D$ and~$x_{u,\rho}$ to~$M$.
        The chosen edges are a matching: each component and each slot
        is used once, and two ordinary edges incident with the same
        $x_{u,\rho}$ would correspond to two transmissions by~$u$ in
        round~$\rho$ in~$P$.  Hence $M$ is feasible and has weight
        $2|Q|+(|\mathcal{D}|-|Q|)=|\mathcal{D}|+|Q|$.
    \end{claimproof}

    \begin{claim}
        If an extendable template exists, then $b(G,s)\le t$.
    \end{claim}
    \begin{claimproof}
        Let a~consistent template and a~matching $M$ of weight at least
        $|\mathcal{D}|+|Q|$ be given.  Let $m_2$ be the number of weight-$2$
        edges in~$M$.  Since each such edge is incident with a~distinct slot,
        $m_2\le |Q|$.  Since $M$ is a~matching, the
        number of weight-$1$ edges is at most $|\mathcal{D}|-m_2$.  Therefore
        \[
            w(M)\le 2m_2+|\mathcal{D}|-m_2
            =|\mathcal{D}|+m_2\le |\mathcal{D}|+|Q|.
        \]
        Equality must hold throughout.  Thus every component is matched, every
        slot is matched by a weight-$2$ edge, and every remaining component is
        matched to a unique vertex $x_{u,\rho}$.

        Realize each slot by its matched component, using a~witnessing local protocol for that realization.
        For every remaining component $D$ matched to $x_{u,\rho}$, choose a neighbor
        $v \in N_G(u)\cap V(D)$ such that $\rho + b(D,v) \le t$, let $u$ inform~$v$ at
        round~$\rho$, and run an optimal broadcast inside $D$ from~$v$.
        The inequality defining the edge to $x_{u,\rho}$ guarantees completion by round~$t$.

        All transmissions are along edges of~$G$: this follows from consistency
        for arcs inside~$S$, from the realizations for slot components, and
        from the definition of the edges to the vertices~$x_{u,\rho}$.  No
        vertex of~$S$ transmits before it is informed or twice in the same
        round: consistency handles the template arcs, while each vertex
        $x_{u,\rho}$ represents a free round of~$u$ and is incident with at
        most one edge of the matching~$M$.
        The local protocols give the same guarantees inside the components, and different components are disjoint.

        Finally, every vertex is informed.  Indeed, every vertex of
        $S\setminus\{s\}$ has the unique incoming transmission prescribed by
        the template, every slot is realized by a matched component, and every
        other component receives one ordinary entry from~$S$.  Following
        incoming transmissions backwards strictly decreases time, and therefore
        reaches the source~$s$.
        Hence this is a~broadcast protocol of length at most~$t$, and
        $b(G,s)\le t$.
    \end{claimproof}

    The algorithm accepts if at least one consistent template is extendable.
    Enumerating the templates and testing all component-slot compatibilities
    takes time $k^{\bO(k^2)}n^{\bO(1)} = \vi^{\bO(\vi^2)}n^{\bO(1)}$,
    and each extension test is polynomial.
\end{proof}

\section{Distance to Clique}\label{sec:dtc}

A set $S \subseteq V(G)$ is a \emph{modulator to a clique} if $G-S$ is a~clique.
The \emph{distance to clique} of a~graph~$G$ is the minimum size of such a~set.
Equivalently, a~modulator to a clique in~$G$ is a vertex cover of the complement~$\overline{G}$.
Thus a~minimum modulator can be computed in time $2^{\bO(k)}n^{\bO(1)}$~\cite{tcs/ChenKX10,stacs/HarrisN24}.
Egami et al.~\cite{mfcs/EgamiGHKLMNOV025} proved that {\TB} is
fixed-parameter tractable parameterized by distance to clique~$k$,
presenting an algorithm of running time $2^{\bO(k^2)}n^{\bO(1)}$.
In this section we improve this dependence to $k^{\bO(k)}n^{\bO(1)}$.

The reason this parameter admits a compact description of a protocol is that the clique side is highly symmetric.
Once one clique vertex is informed, the message can spread inside the clique without using the modulator again.
Thus the only clique vertices that need to be represented explicitly are those that send the message back to the modulator.
After adding the source to the modulator, there are at most $|S|-1$ such vertices.
We represent them by slots in a template, use a {\bMatching}
instance to choose the actual clique vertices realizing these slots,
and then complete the broadcast inside the clique greedily.

We use the following upper bound of Egami et al.~\cite{mfcs/EgamiGHKLMNOV025}
to restrict the length of the optimal protocol.

\begin{lemma}[{\cite[Lemma 10]{mfcs/EgamiGHKLMNOV025}}]\label{lem:dtc:upper-bound}
    Let $G$ be a connected graph, and let $S \subseteq V(G)$ be such
    that $G-S$ is a clique of size $N \ge 1$.
    Then $b(G,s) \le |S| + \ceil{\log_2 N}$ for every source $s \in V(G)$.
\end{lemma}

We now present the improved algorithm parameterized by distance to clique.

\begin{theorem}\label{thm:algo:dc}   
    {\TB} can be solved in time $k^{\bO(k)}n^{\bO(1)}$,
    where $k$ is the distance to clique of the input graph.
\end{theorem}                        

\begin{proof}
    Let $(G,s,t)$ be an instance of {\TB}, where $G=(V,E)$ is connected and $|V|=n$.
    We compute a minimum modulator $S_0$ to a clique in time $2^{\bO(k)}n^{\bO(1)}$.
    Let $S := S_0 \cup \{s\}$, $h := |S| \le k+1$, and $C:= V \setminus S$ where $N := |C|$.
    It holds that $s \in S$ and $G[C]=G-S$ is a clique.
    Assume that $t < h + \ceil{\log_2 N}$, since otherwise the answer is trivially yes by \cref{lem:dtc:upper-bound}.

    \proofsubparagraph{Template.}
    Consider an optimal protocol $P$ and let~$T$ be its broadcast tree.
    We call a clique vertex \emph{marked} if it has a child in~$S$ in~$T$.
    Assigning to each marked clique vertex one of its children in
    $S\setminus\{s\}$ gives an injection into $S\setminus\{s\}$,
    because each vertex has a unique parent in~$T$.
    Hence there are at most $h-1$ marked clique vertices.

    We enumerate all possible guesses for this bounded part of the protocol.
    For each $w\in\{0,\ldots,h-1\}$, we introduce a set $Q$ of $w$ \emph{marked slots}.
    A marked slot is a placeholder for a marked vertex of~$C$;
    it will later be realized by a distinct clique vertex.
    The named vertices of the template are the vertices of~$S$.

    We define a template $\mathcal{T}=(H,s,S,Q,\lambda)$, with
    $H=(S\cup Q\cup\{\phi\},A)$, by a parent map, where $\phi$ is an auxiliary vertex.
    For every $u\in S\setminus\{s\}$, we guess a parent $p(u)\in (S\cup Q)\setminus\{u\}$.
    For every slot $q\in Q$, we guess $p(q) \in S\cup\{\phi\}$,
    where~$\phi$ means that the clique vertex realizing~$q$ is informed from
    inside the clique.
    We add the arc $p(z)z$ for every $z\in (S\setminus\{s\})\cup Q$.
    Thus every vertex of $(S\setminus\{s\})\cup Q$ has a unique incoming arc,
    while~$\phi$ has none.

    We also guess the label of every arc of~$H$ from~$[t]$, as well as a \emph{load vector}
    $\ell \colon S \to \{0,\ldots,t\}$, where $\ell(u)$ is the number of unmarked
    clique vertices informed directly by~$u$.
    The marked slots informed by~$u$ are already visible in the template.
    For every $u\in S$, let
    \[
        \children_S(u):=\setdef{u'\in S \setminus \{s\}}{p(u')=u}
        \quad\text{and}\quad
        \children_Q(u):=\setdef{q\in Q}{p(q)=u}.
    \]

    For the rest of the proof, the guessed object is the pair $(\mathcal{T} = (H,s,S,Q,\lambda),\ell)$,
    where, for notational convenience, we write $\tau=\tau_\lambda$.
    We call this pair \emph{consistent} if $\mathcal{T}$ satisfies the common consistency checks of
    \cref{sec:templates} and the following additional conditions:
    \begin{itemize}
        \item every slot has at least one child in~$S$;
        \item the outgoing arcs of each slot have pairwise distinct labels;
        \item for every $u\in S$,
        $|\children_S(u)|+|\children_Q(u)|+\ell(u) \le t-\tau(u)$.
    \end{itemize}
    The last condition says that, after $u$ is informed, there are enough rounds
    for the represented outgoing arcs of~$u$ and for the $\ell(u)$ additional
    entries from~$u$ into the clique.

    \begin{claim}\label{clm:dtc:templates}
        The number of consistent pairs $(\mathcal{T},\ell)$ considered is
        $h^{\bO(h)}n^{\bO(1)}$.
    \end{claim}

    \begin{claimproof}
        Fix $|Q|=w\le h-1$.  There are at most $(2h)^h$ choices for the parents
        of the vertices in $S\setminus\{s\}$, and at most $(h+1)^h$ choices for
        the parents of the slots.  Since
        $t< h+\ceil{\log_2 n}$, the labels and the load vector together have
        $(h+\log n)^{\bO(h)}$ possibilities in total.
        Hence the number of guesses is
        \[
            h^{\bO(h)}(h+\log n)^{\bO(h)}.
        \]
        If $\log n\le h^2$, this is $h^{\bO(h)}$.
        Otherwise $h<\sqrt{\log n}$, and the factor
        $(h+\log n)^{\bO(h)}$ is bounded by $n^{\bO(1)}$.
    \end{claimproof}

    \proofsubparagraph{Realizing the template.}
    Fix a consistent pair $(\mathcal{T},\ell)$.  For a slot $q\in Q$, let
    \[
        R(q):=\setdef{u\in S \setminus \{s\}}{p(u)=q}
        \cup
        \begin{cases}
            \{p(q)\}, & \text{if }p(q)\in S,\\
            \varnothing, & \text{if }p(q)=\phi.
        \end{cases}
    \]
    A clique vertex can realize~$q$ only if it is adjacent to every vertex of
    $R(q)$: it has to inform the children of~$q$ in~$S$, and,
    when $p(q)\in S$, it has to receive the message from~$p(q)$.
    Let $C_q:=\setdef{v\in C}{R(q)\subseteq N_G(v)}$ denote the set of clique
    vertices that can realize~$q$.

    We use one edge-weighted {\bMatching} instance for two tasks:
    realizing the slots by distinct clique vertices, and choosing the
    $\ell(u)$ unmarked clique vertices that are informed directly by each
    $u\in S$.
    Its vertex set contains one vertex for every $v\in C$, one vertex $w_q$ for
    every slot $q\in Q$, and one vertex for every $u\in S$ with $\ell(u)>0$.
    The capacities are
    \[
        \beta(v)=1\quad(v\in C),\qquad
        \beta(w_q)=1\quad(q\in Q),\qquad
        \beta(u)=\ell(u)\quad(u\in S,\ \ell(u)>0).
    \]
    We add an edge of weight~$2$ between $w_q$ and every vertex of~$C_q$.
    We also add an edge of weight~$1$ between $u\in S$ and $v\in C$ whenever
    $uv\in E(G)$ and $\ell(u)>0$.

    A feasible $b$-matching of weight
    \[
        2|Q|+\sum_{u\in S}\ell(u)
    \]
    is called a \emph{realizing matching} for the pair $(\mathcal{T},\ell)$.
    This is the maximum possible weight: the slot vertices contribute at most
    $2|Q|$, and the vertices of~$S$ contribute at most $\sum_{u\in S}\ell(u)$.
    Hence a realizing matching saturates every slot vertex $w_q$ and uses
    exactly $\ell(u)$ edges incident with each $u\in S$.
    Equivalently, it chooses vertices $v_q\in C_q$ for all slots $q\in Q$
    and sets $B_u\subseteq N_G(u)\cap C$ with $|B_u|=\ell(u)$ for every
    $u\in S$, such that all vertices in
    $\setdef{v_q}{q\in Q}\cup\bigcup_{u\in S}B_u$ are distinct.
    We discard the pair if no realizing matching exists.
    This test is polynomial-time, since maximum-weight {\bMatching} is
    polynomial-time solvable~\cite{book/Schrijver03}.

    \proofsubparagraph{Extending a realization.}
    Fix a consistent pair $(\mathcal{T},\ell)$ and any realizing matching~$M$.
    For every slot $q\in Q$, let $v_q\in C$ be the vertex matched to~$w_q$.
    For every $u\in S$, let \[B_u:=\setdef{v\in C}{uv\in M}\] if $\ell(u)>0$, and
    let $B_u:=\varnothing$ otherwise.  Then $|B_u|=\ell(u)$.
    Since $C$ is a clique, the greedy completion below depends on~$M$ only
    through the fact that every slot is realized and that $|B_u|=\ell(u)$ for
    every $u\in S$; the actual identities of the chosen clique vertices are
    irrelevant.

    We first schedule all template arcs not leaving~$\phi$ at their labels:
    arcs inside~$S$ use the corresponding edges of~$G$,
    an arc $u q$ is realized as $u v_q$, and an arc $q u$ is realized as
    $v_q u$.
    For every $u\in S$, let
    \[
        F_u:=\setdef{\rho\in[\tau(u)+1,t]}
        {\text{no arc of }H\text{ leaves }u\text{ with label }\rho}.
    \]
    By consistency, $|F_u|\ge \ell(u)$.
    We schedule the vertices of~$B_u$ in the $\ell(u)$ earliest rounds of~$F_u$.

    It remains to complete the broadcast inside the clique.
    All clique vertices except the vertices $v_q$ realizing slots and the
    vertices in the sets~$B_u$ are anonymous.
    Slots with parent in~$S$ and vertices in the sets~$B_u$ have fixed incoming transmissions from~$S$.
    Slots with parent~$\phi$ must instead be informed from inside the clique at their incoming label.
    Starting with no informed clique vertex, we proceed round by round.
    At the beginning of round~$\rho$, let $N_\rho$ be the number of informed clique vertices.
    Let \[D_\rho:=\setdef{q\in Q}{\text{there is an arc }e=(q,u)\in A\text{ with }u\in S\text{ and }\lambda(e)=\rho}\]
    be the slots that must inform a vertex of~$S$ in this round.  These slots
    are already committed to a transmission to~$S$, and therefore are not
    available for clique-internal transmissions in round~$\rho$.
    Let \[E_\rho:=\setdef{q\in Q}{p(q)=\phi \text{ and } \tau(q)=\rho}\]
    be the slots that must be informed from inside the clique in this round.
    If some slot of $D_\rho$ is not informed at the beginning of the round, or if
    $N_\rho-|D_\rho|<|E_\rho|$, the test fails.
    Otherwise, there are $N_\rho-|D_\rho|$ informed clique vertices available
    for clique-internal transmissions.  We use $|E_\rho|$ of them to inform
    the slots of~$E_\rho$, and the remaining available transmissions to inform
    as many still uninformed anonymous clique vertices as possible.
    At the end of the round, we mark as informed all clique vertices reached in
    this round: the slots of~$E_\rho$, the anonymous clique vertices just
    chosen greedily, and the clique vertices whose fixed incoming transmission
    from~$S$ occurs in round~$\rho$.  The latter are the slots with parent
    in~$S$ and the vertices of the sets~$B_u$ scheduled in round~$\rho$.

    Since the test only uses the fixed time-stamps and the numbers
    $|B_u|=\ell(u)$, it is independent of the particular realizing matching.
    We call the pair $(\mathcal{T},\ell)$ \emph{extendable} if it admits a
    realizing matching and this greedy test succeeds.

    \begin{claim}\label{clm:dtc:complete}
        If $b(G,s)\le t$, then there exists a consistent pair
        $(\mathcal{T},\ell)$ that is extendable.
    \end{claim}

    \begin{claimproof}
        Suppose that $b(G,s)\le t$, and fix a protocol~$P$ of length at
        most~$t$ with broadcast tree~$T$.  Create one slot for every marked
        clique vertex of~$P$.  For every vertex of $S\setminus\{s\}$, record
        its parent in~$T$; if the parent is a marked clique vertex, record the
        corresponding slot.  For every slot, record its parent if that parent
        lies in~$S$, and record~$\phi$ otherwise.  The labels of the incoming
        arcs are the reception times read directly from~$P$.  Finally, set
        $\ell(u)$ to be the number of unmarked
        clique vertices informed directly by~$u$ in~$P$.  This gives one of
        the pairs enumerated above, and it is consistent.

        We now build a realizing matching.  For every slot $q$, match $w_q$ to the actual
        clique vertex realizing~$q$ in~$P$.  This edge exists because all
        adjacencies to vertices of~$S$ required by the template occur in~$P$.
        For every clique vertex that is informed by some $u\in S$ and is not
        represented by a slot, add the edge between $u$ and that clique vertex.
        The chosen edges respect all capacities, because a clique vertex appears
        only once in the broadcast tree, and each vertex $u$ informs exactly
        $\ell(u)$ unmarked clique vertices.  The weight is
        $2|Q|+\sum_{u\in S}\ell(u)$, so this is a realizing matching.

        It remains to see that the greedy extension test succeeds.  Compare it
        with the restriction of~$P$ to the clique.  The template fixes the same
        slot-to-$S$ transmissions as~$P$, and every slot with parent~$\phi$ must
        be informed in~$P$ at its guessed round.
        For each $u\in S$, the extension schedules exactly as many unmarked
        clique vertices from~$u$ as~$P$ does, and it uses the earliest free
        rounds of~$u$ for them.  Hence these entries occur no later than the
        corresponding entries of~$P$.
        Since unmarked clique vertices are interchangeable, an induction over
        the rounds shows that, just before each round, the greedy extension has
        informed at least as many clique vertices as~$P$ has, while the same
        slots are busy sending to~$S$.  Thus whenever~$P$ can inform the slots
        with $p(q)=\phi$, the greedy extension can do so as well; after that it
        uses all remaining clique-internal capacity.  Since~$P$ informs all
        clique vertices by round~$t$, the greedy extension also does.  Hence the
        pair $(\mathcal{T},\ell)$ is extendable.
    \end{claimproof}

    \begin{claim}\label{clm:dtc:sound}
        If there exists a consistent pair $(\mathcal{T},\ell)$ that is extendable,
        then $b(G,s)\le t$.
    \end{claim}

    \begin{claimproof}
        Suppose that a consistent pair $(\mathcal{T},\ell)$ is extendable, and fix a realizing
        matching and greedy extension witnessing this.  The construction above
        gives a protocol of length at most~$t$.  All transmissions are along
        edges of~$G$: this follows from consistency for arcs inside~$S$, from
        the definition of the sets~$C_q$ for arcs between~$S$ and slots, from the
        matching edges between~$S$ and~$C$, and from the fact that~$C$ is a
        clique for all clique-internal transmissions, including those used to
        satisfy arcs leaving~$\phi$.

        No vertex transmits before being informed or twice in the same round.
        For vertices of~$S$ this follows from consistency and from the choice of
        the free rounds~$F_u$; for clique vertices it is enforced by the greedy
        simulation, which treats slot-to-$S$ transmissions as occupied rounds.
        Finally, every vertex is informed by construction.  Following incoming
        transmissions backward strictly decreases time, and therefore reaches
        the source~$s$.  Hence $b(G,s)\le t$.
    \end{claimproof}

    By the two preceding claims, testing all pairs $(\mathcal{T},\ell)$ decides whether
    $b(G,s)\le t$.
    By \cref{clm:dtc:templates}, the number of pairs is
    $h^{\bO(h)}n^{\bO(1)}=k^{\bO(k)}n^{\bO(1)}$, since $h\le k+1$.
    For each pair we solve one polynomial-size {\bMatching} instance to
    test whether it admits a realizing matching and, if so, run the polynomial-time
    greedy extension test.
    Hence the total running time is $k^{\bO(k)}n^{\bO(1)}$.
\end{proof}

\section{Conclusion}

We gave faster FPT algorithms for \TB\ under various parameterizations, leading to
faster and simpler algorithms that make use of maximum-weight \textsc{$b$-Matching}.
As an interesting open question, we ask whether the problem admits
an algorithm running in time $2^{\bO(\vc)} n^{\bO(1)}$ or perhaps whether a $\vc^{o(\vc)}$ dependence can be excluded under the ETH.
Given the problem's hardness, another interesting question is which other structural parameterizations
render it FPT; is that the case for cutwidth, or perhaps bandwidth?

\bibliography{refs}

\end{document}